\documentclass[a4paper,USenglish,nomarks,final]{dmtcs-episciences}

\usepackage[utf8]{inputenc}

\usepackage{latexsym,amsthm,amsmath,amssymb}

\usepackage[round]{natbib}
\usepackage{type1ec}
\usepackage[T1]{fontenc}
\usepackage{bbm}
\usepackage{mathtools}

\usepackage{amsfonts}
\usepackage{pifont}
\usepackage{url}
\usepackage{multirow}
\usepackage{multicol}
\usepackage{pbox}
\usepackage{fancybox}
\usepackage{colortbl}
\usepackage{soul}
\usepackage{color}
\usepackage{tikz}
\usepackage{tkz-berge}
\usepackage{rotating}
\usepackage{pdflscape}
\usepackage{afterpage}
\usepackage{comment}
\usepackage{etoolbox}
\usepackage{environ}
\usepackage{dirtytalk}

\usepackage{algorithm}
\usepackage{algorithmic}

\hypersetup{colorlinks = true, urlcolor = blue, citecolor = blue}

\usetikzlibrary{decorations,arrows,shapes}

\newcommand{\td}[1]{\mathbb{#1}}
 
\newtheorem{definition}{Definition}
\newtheorem{theorem}{Theorem}
\newtheorem{lemma}{Lemma}
\newtheorem{corollary}{Corollary}

\usepackage{complexity}

\newcommand{\bigO}[1]{\mathcal{O}\!\left(#1\right)}
\newcommand{\bigOs}[1]{\mathcal{O}^*\!\left(#1\right)}

\newcommand{\pname}[1]{\textsc{#1}}

\newcommand{\problem}[3]{{\centering\fbox{\pbox{\textwidth}{\pname{#1}\\\textit{Instance}: #2\\\textit{Question}: #3}}}}
\newcommand{\pproblem}[4]{{\centering\fbox{\pbox{\textwidth}{\pname{#1}\\\textit{Instance}: #2\\\textit{Parameter}: #3\\\textit{Question}: #4}}}}

\newcommand{\inners}{1.2pt}
\newcommand{\outers}{1pt}
\newcommand{\gscale}{0.7}

\newclass{\Hard}{hard}
\newclass{\Hness}{hardness}
\newcommand{\NPH}{\NP\text{-}\Hard}
\newclass{\Complete}{complete}
\newclass{\Cness}{completeness}
\newcommand{\NPc}{\NP\text{-}\Complete}

\newfunc{\dist}{dist}
\newfunc{\girth}{girth}
\newfunc{\nd}{nd}
\newfunc{\YES}{YES}
\newfunc{\NOi}{NO}
\newfunc{\ff}{ff}
\newfunc{\cf}{cf}
\newfunc{\tw}{tw}

\newtoggle{proof_toggle}
\toggletrue{proof_toggle}

\NewEnviron{tproof} {
    \iftoggle{proof_toggle} {%
        \begin{proof}\BODY\end{proof}}
    {%
    }
}

\newcommand{\ceil}[1]{\left\lceil#1\right\rceil}
\newcommand{\floor}[1]{\left\lfloor#1\right\rfloor}

\title[Parameterized Complexity of Equitable Coloring]{Parameterized Complexity of Equitable Coloring\thanks{Work supported by Brazilian projects CNPq 311013/2015-5, CNPq Universal 421660/2016-3, FAPEMIG, and Coordena\c c\~ao de Aperfei\c coamento de Pessoal de N\'ivel Superior - Brasil (CAPES) - Finance Code 001.}}

\author{Guilherme de C. M. Gomes
        \and Carlos V. G. C. Lima
        \and Vin\'icius F. dos Santos\thanks{Email addresses: gcm.gomes@dcc.ufmg.br (Gomes), carloslima@dcc.ufmg.br (Lima), viniciussantos@dcc.ufmg.br (Santos).}}
\affiliation{Departamento de Ci\^encia da Computa\c c\~ao, Universidade Federal de Minas Gerais, Belo Horizonte, Brazil}


\received{2018-11-1}
\revised{2019-4-15}
\accepted{2019-5-4}

\begin{document}

\publicationdetails{21}{2019}{1}{8}{4948}

\maketitle

\begin{abstract}
    A graph on $n$ vertices is equitably $k$-colorable if it is $k$-colorable and every color is used either 
    $\left\lfloor n/k \right\rfloor$ or 
    $\left\lceil n/k \right\rceil$ times.
    Such a problem appears to be considerably harder than vertex coloring, being 
    $\mathsf{NP\text{-}complete}$ even for cographs and interval graphs.
    In this work, we prove that it is 
    $\mathsf{W[1]\text{-}hard}$ for block graphs and for disjoint union of split graphs when parameterized by the number of colors; and 
    $\mathsf{W[1]\text{-}hard}$ for $K_{1,4}$-free interval graphs when parameterized by treewidth, number of colors and maximum degree, generalizing a result by Fellows et al. (2014) through a much simpler reduction.
    Using a previous result due to Dominique de Werra (1985), we establish a dichotomy for the complexity of equitable coloring of chordal graphs based on the size of the largest induced star.
    Finally, we show that 
    \textsc{equitable coloring} is 
    $\mathsf{FPT}$ when parameterized by the treewidth of the complement graph.
\end{abstract}

\section{Introduction}

\pname{equitable coloring} is a variant of the classical \pname{vertex coloring} problem, where we not only want to partition an $n$ vertex graph into $k$ independent sets, but also that each of these sets has either $\floor{n/k}$ or $\ceil{n/k}$ vertices.
The smallest integer $k$ for which $G$ admits an equitable $k$-coloring is called the \emph{equitable chromatic number} of $G$.

An extensive survey was conducted by \cite{equitable_survey}, where many of the results on \pname{equitable coloring} of the last 50 years were assembled.
Most of them, however, are upper bounds on the equitable chromatic number.
Such bounds are known for:
bipartite graphs,
trees,
split graphs,
planar graphs,
outerplanar graphs,
low degeneracy graphs,
Kneser graphs,
interval graphs,
random graphs
and some forms of graph products.

Almost all complexity results for \pname{equitable coloring} arise from a related problem, known as \pname{bounded coloring}, an observation given by~\cite{equitable_treewidth}.
On \pname{bounded coloring}, we ask that the size of the independent sets be bounded by an integer $\ell$, which is not necessarily a function on $k$ or $n$.
Among the known results for \pname{bounded coloring}, we have that the problem is solvable in polynomial time for:
split graphs~\citep{equitable_split},
complements of interval graphs~\citep{graph_partitioning1}, complements of bipartite graphs~\citep{graph_partitioning1},
forests~\citep{mutual_exclusion_scheduling}, and
trees~\citep{mutual_exclusion_scheduling, equitable_trees}. 
While the algorithm for \pname{bounded coloring} on trees was first presented by~\cite{mutual_exclusion_scheduling}, \cite{equitable_trees} compute the minimum number of colors required to color a tree using at most $\ell$ colors using a novel characterization.
For cographs, there is a polynomial-time algorithm when the number of colors $k$ is fixed, otherwise the problem is $\NPc$~\citep{graph_partitioning1}; the same is also valid for bipartite graphs and interval graphs~\citep{graph_partitioning1}.
A consequence of the difficulty of \pname{bounded coloring} for cographs is the difficulty of the problem for graphs of bounded cliquewidth.
In complements of comparability graphs, even if we fix $\ell \geq 3$, \pname{bounded coloring} remains $\NPH$~\citep{chain_antichain}, implying that the problem is $\mathsf{para}\NPH$ for the parameter $\ell$.
\cite{colorful_treewidth} show that \pname{equitable coloring} parameterized by treewidth and number of colors is $\W[1]$-$\Hard$ and, in~\citep{equitable_treewidth}, an $\XP$ algorithm parameterized by treewidth is given for both \pname{equitable coloring} and \pname{bounded coloring}.

In this work, we perform a series of reductions proving that \pname{equitable coloring} is $\W[1]$-$\Hard$ for different subclasses of chordal graphs.
In particular, we show that the problem parameterized by the number of colors is $\W[1]$-$\Hard$ for block graphs and for the disjoint union of split graphs.
Moreover, the problem remains $\W[1]$-$\Hard$ for $K_{1,4}$-free interval graphs even if we parameterize it by treewidth, number of colors and maximum degree.
These results generalize the proof given by~\cite{colorful_treewidth} that \pname{equitable coloring} is $\W[1]$-$\Hard$ when parameterized by treewidth and number of colors.
A result given by \cite{claw_free_de_werra} guarantees that every $K_{1,3}$-free graph can be equitable $k$-colored if $k$ is at least the chromatic number.
Since \pname{vertex coloring} can be solved in polynomial time on chordal graphs, we trivially have a polynomial-time algorithm for \pname{equitable coloring} of $K_{1,3}$-free chordal graphs.
This allows us to establish a dichotomy for the computational complexity of \pname{equitable coloring} of chordal graphs based on the size of the largest induced star.
\section{Preliminaries}

We used standard graph theory notation.
Define $[k] = \{1,\dots, k\}$ and $2^S$ the \emph{powerset} of $S$.
A \emph{$k$-coloring} $\varphi$ of a graph $G$ is a function $\varphi: V(G) \mapsto~[k]$.
Alternatively, a $k$-coloring is a $k$-partition $V(G) \sim \{\varphi_1, \dots, \varphi_k\}$ such that $\varphi_i = \{u \in V(G) \mid \varphi(u) = i\}$.
A set $X \subseteq V(G)$ is \emph{monochromatic} if~$\left|\bigcup_{u \in X} \{\varphi(u)\}\right| = 1$.
A $k$-coloring is said to be \emph{equitable} if, for every $i \in [k]$, $\floor{n/k} \leq |\varphi_i| \leq \ceil{n/k}$.
A $k$-coloring of $G$ is \emph{proper} if no edge of $G$ is monochromatic, that is, if $\varphi_i$ is an independent set for every $i \in [k]$.
Unless stated, all colorings are proper.

The \emph{disjoint union}, or simply \emph{union}, of two graphs $G \cup H$ is a graph such that~$V(G \cup H) = V(G) \cup V(H)$ and $E(G \cup H) = E(G) \cup E(H)$.
The \emph{join} of two graphs~$G \oplus H$ is the graph given by~$V(G \oplus H) = V(G) \cup V(H)$ and $E(G \oplus H) = E(G) \cup E(H) \cup \{uv \mid u \in V(G), v \in V(H)\}$.
A \textit{$k$-connected component} is a connected component such that $k$ vertices must be removed for it to become disconnected.
A component is \textit{biconnected} if $k=2$.
A graph is a \emph{block graph} if and only if every biconnected component is a clique;
it is a \emph{split graph} if and only if $V(G)$ can be partitioned in a clique and an independent set.
The \emph{length} of a path $P_n$ on $n$ vertices is the number of edges it contains, that is, $n-1$.
The \emph{diameter} of a graph is the length of the longest minimum path between any two vertices of the graph.

A \textit{tree decomposition} of a graph $G$ is defined as $\td{T} = \left(T, \mathcal{B} = \{B_j \mid j \in V(T)\}\right)$, where $T$ is a tree and $\mathcal{B} \subseteq 2^{V(G)}$ is a family where: $\bigcup_{B_j \in \mathcal{B}} B_j = V(G)$;
for every edge $uv \in E(G)$ there is some~$B_j$ such that $\{u,v\} \subseteq B_j$;
for every $i,j,q \in V(T)$, if $q$ is in the path between $i$ and $j$ in $T$, then $B_i \cap B_j \subseteq B_q$.
Each $B_j \in \mathcal{B}$ is called a \emph{bag} of the tree decomposition.
The \emph{width} of a tree decomposition is defined as the size of a largest bag minus one.
The \emph{treewidth} $\tw(G)$ of a graph $G$ is the smallest width among all valid tree decompositions of $G$~\citep{downey_fellows}.
If $\td{T}$ is a rooted tree, by $G_x$ we will denote the subgraph of $G$ induced by the vertices contained in any bag that belongs to the subtree of $\td{T}$ rooted at bag $x$.
An algorithmically useful property of tree decompositions is the existence of a so called \emph{nice tree decompositions} of width $\tw(G)$.

\begin{definition}{Nice tree decomposition}
    A tree decomposition $\td{T}$ of $G$ is said to be \emph{nice} if it is a tree rooted at, say, the empty bag $r(T)$ and each of its bags is from one of the following four types:
    \begin{enumerate}
        \item \emph{Leaf node}: a leaf $x$ of $\td{T}$ with $B_x = \emptyset$.
        \item \emph{Introduce node}: an inner bag $x$ of $\td{T}$ with one child $y$ such that $B_x \setminus B_y = \{u\}$.
        \item \emph{Forget node}: an inner bag $x$ of $\td{T}$ with one child $y$ such that $B_y \setminus B_x = \{u\}$.
        \item \emph{Join node}: an inner bag $x$ of $\td{T}$ with two children $y,z$ such that $B_x = B_y = B_z$.
    \end{enumerate}
\end{definition}
\section{Subclasses of Chordal Graphs}

\begin{figure}[!tb]
    \centering
    \begin{tikzpicture}[scale=\gscale]
        \GraphInit[unit=3,vstyle=Normal]
        \SetVertexNormal[Shape=circle, FillColor=black, MinSize=2pt]
        \tikzset{VertexStyle/.append style = {inner sep = \inners, outer sep = \outers}}
        \SetVertexNoLabel
        \Vertex[x=0,y=0,L={y}, Lpos=270, LabelOut, Ldist=3pt]{y}
        \draw[] (0,0) circle (0.4cm);
        \Vertex[a=-5, d=1.7cm]{x11}
        \Vertex[a=-30, d=2.3cm]{x12}
        \Vertex[a=-55, d=1.7cm]{x13}
        
        \Vertex[a=115, d=1.7cm]{x21}
        \Vertex[a=90, d=2.3cm]{x22}
        \Vertex[a=65, d=1.7cm]{x23}
        
        \Vertex[a=235, d=1.7cm]{x31}
        \Vertex[a=210, d=2.3cm]{x32}
        \Vertex[a=185, d=1.7cm]{x33}
        
        \foreach \x in {1,2,3}
            \foreach \y in {1,2,3}
                \Edge(y)(x\x\y);
        
        \Edge(x11)(x12)
        \Edge(x11)(x13)
        \Edge(x12)(x13)

        \Edge(x21)(x22)
        \Edge(x21)(x23)
        \Edge(x22)(x23)

        \Edge(x31)(x32)
        \Edge(x31)(x33)
        \Edge(x32)(x33)

    \end{tikzpicture}
    \hfill
    \begin{tikzpicture}[scale=\gscale]
        \GraphInit[unit=3,vstyle=Normal]
        \SetVertexNormal[Shape=circle, FillColor=black, MinSize=2pt]
        \tikzset{VertexStyle/.append style = {inner sep = \inners, outer sep = \outers}}
        \SetVertexNoLabel
        \grComplete[RA=1]{3}
        \Vertex[a=60,d=1.5]{x}
        \Vertex[a=180,d=1.5]{y}
        \Vertex[a=300,d=1.5]{z}
        \draw[] (60:1.5) circle (0.4cm);
        \draw[] (180:1.5) circle (0.4cm);
        \draw[] (300:1.5) circle (0.4cm);
        
        \Edge(x)(a0)
        \Edge(x)(a1)
        \Edge(x)(a2)
        
        \Edge(y)(a0)
        \Edge(y)(a1)
        \Edge(y)(a2)
        
        \Edge(z)(a0)
        \Edge(z)(a1)
        \Edge(z)(a2)
        
    \end{tikzpicture}
    \hfill
    \begin{tikzpicture}[scale=\gscale]
        \GraphInit[unit=3,vstyle=Normal]
        \SetVertexNormal[Shape=circle, FillColor=black, MinSize=2pt]
        \tikzset{VertexStyle/.append style = {inner sep = \inners, outer sep = \outers}}
        \SetVertexNoLabel
        \Vertex[x=0,y=0]{y1}
        
        \draw[] (0,0) circle (0.4cm);
        
        \Vertex[x=3,y=0]{y2}
        
        \draw[] (3,0) circle (0.4cm);
        
        \begin{scope}[shift={(1.5cm,0)}, rotate=-90]
            \grComplete[RA=1, prefix=a]{2}
        \end{scope}
        
        \begin{scope}[shift={(0cm,-1.5cm)}]
            \grComplete[RA=1, prefix=c]{2}
        \end{scope}
        
        \begin{scope}[shift={(-1.5cm,0)}, rotate=-90]
            \grComplete[RA=1, prefix=b]{2}
        \end{scope}
        
        \begin{scope}[shift={(3cm,-1.5cm)}]
            \grComplete[RA=1, prefix=d]{2}
        \end{scope}
        
        \Edges(a0,y1,a1)
        \Edges(b0,y1,b1)
        \Edges(c0,y1,c1)
        
        \Edges(a0,y2,a1)
        \Edges(d0,y2,d1)

    \end{tikzpicture}
    \hfill
    
    \caption{A $(2,4)$-flower, a $(2,4)$-antiflower, and a $(2,2)$-trem.}
    \label{fig:flower}
\end{figure}
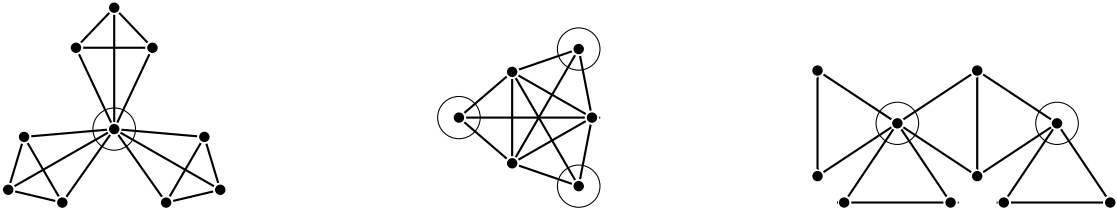

\begin{figure}[!tb]
    \centering
    \begin{tikzpicture}[scale=\gscale]
        \begin{scope}[rotate=36,shift={(0cm, 5cm)}]
            \GraphInit[unit=3,vstyle=Normal]
            \SetVertexNormal[Shape=circle, FillColor=black, MinSize=2pt]
            \tikzset{VertexStyle/.append style = {inner sep = \inners, outer sep = \outers}}
            \SetVertexNoLabel
            \Vertex[Math, x=0,y=0,L={y_1}, Lpos=290, LabelOut, Ldist=3pt]{1y}
            \Vertex[a=-5, d=1.7cm]{1x11}
            \Vertex[a=-30, d=2.3cm]{1x12}
            \Vertex[a=-55, d=1.7cm]{1x13}
            
            \Vertex[a=115, d=1.7cm]{1x21}
            \Vertex[a=90, d=2.3cm]{1x22}
            \Vertex[a=65, d=1.7cm]{1x23}
            
            \Vertex[a=235, d=1.7cm]{1x31}
            \Vertex[a=210, d=2.3cm]{1x32}
            \Vertex[a=185, d=1.7cm]{1x33}
            
            \foreach \x in {1,2,3}
                \foreach \y in {1,2,3}
                    \Edge(1y)(1x\x\y);
            
            \Edge(1x11)(1x12)
            \Edge(1x11)(1x13)
            \Edge(1x12)(1x13)

            \Edge(1x21)(1x22)
            \Edge(1x21)(1x23)
            \Edge(1x22)(1x23)

            \Edge(1x31)(1x32)
            \Edge(1x31)(1x33)
            \Edge(1x32)(1x33)
        \end{scope}
        \begin{scope}[rotate=-36,shift={(0cm, 5cm)}]
            \GraphInit[unit=3,vstyle=Normal]
            \SetVertexNormal[Shape=circle, FillColor=black, MinSize=2pt]
            \tikzset{VertexStyle/.append style = {inner sep = \inners, outer sep = \outers}}
            \SetVertexNoLabel
            \Vertex[Math, x=0,y=0,L={y_1}, Lpos=290, LabelOut, Ldist=3pt]{2y}
            \Vertex[a=-5, d=1.7cm]{2x11}
            \Vertex[a=-30, d=2.3cm]{2x12}
            \Vertex[a=-55, d=1.7cm]{2x13}
            
            \Vertex[a=115, d=1.7cm]{2x21}
            \Vertex[a=90, d=2.3cm]{2x22}
            \Vertex[a=65, d=1.7cm]{2x23}
            
            \Vertex[a=235, d=1.7cm]{2x31}
            \Vertex[a=210, d=2.3cm]{2x32}
            \Vertex[a=185, d=1.7cm]{2x33}
            
            \foreach \x in {1,2,3}
                \foreach \y in {1,2,3}
                    \Edge(2y)(2x\x\y);
            
            \Edge(2x11)(2x12)
            \Edge(2x11)(2x13)
            \Edge(2x12)(2x13)

            \Edge(2x21)(2x22)
            \Edge(2x21)(2x23)
            \Edge(2x22)(2x23)

            \Edge(2x31)(2x32)
            \Edge(2x31)(2x33)
            \Edge(2x32)(2x33)
        \end{scope}
        \begin{scope}[rotate=-108,shift={(0cm, 5cm)}]
            \GraphInit[unit=3,vstyle=Normal]
            \SetVertexNormal[Shape=circle, FillColor=black, MinSize=2pt]
            \tikzset{VertexStyle/.append style = {inner sep = \inners, outer sep = \outers}}
            \SetVertexNoLabel
            \Vertex[Math, x=0,y=0,L={y_1}, Lpos=290, LabelOut, Ldist=3pt]{3y}
            \Vertex[a=-5, d=1.7cm]{3x11}
            \Vertex[a=-30, d=2.3cm]{3x12}
            \Vertex[a=-55, d=1.7cm]{3x13}
            
            \Vertex[a=115, d=1.7cm]{3x21}
            \Vertex[a=90, d=2.3cm]{3x22}
            \Vertex[a=65, d=1.7cm]{3x23}
            
            \Vertex[a=235, d=1.7cm]{3x31}
            \Vertex[a=210, d=2.3cm]{3x32}
            \Vertex[a=185, d=1.7cm]{3x33}
            
            \foreach \x in {1,2,3}
                \foreach \y in {1,2,3}
                    \Edge(3y)(3x\x\y);
            
            \Edge(3x11)(3x12)
            \Edge(3x11)(3x13)
            \Edge(3x12)(3x13)

            \Edge(3x21)(3x22)
            \Edge(3x21)(3x23)
            \Edge(3x22)(3x23)

            \Edge(3x31)(3x32)
            \Edge(3x31)(3x33)
            \Edge(3x32)(3x33)
        \end{scope}
        \begin{scope}[rotate=108,shift={(0cm, 5cm)}]
            \GraphInit[unit=3,vstyle=Normal]
            \SetVertexNormal[Shape=circle, FillColor=black, MinSize=2pt]
            \tikzset{VertexStyle/.append style = {inner sep = \inners, outer sep = \outers}}
            \SetVertexNoLabel
            \Vertex[Math, x=0,y=0,L={y_1}, Lpos=290, LabelOut, Ldist=3pt]{4y}
            \Vertex[a=-5, d=1.7cm]{4x11}
            \Vertex[a=-30, d=2.3cm]{4x12}
            \Vertex[a=-55, d=1.7cm]{4x13}
            
            \Vertex[a=115, d=1.7cm]{4x21}
            \Vertex[a=90, d=2.3cm]{4x22}
            \Vertex[a=65, d=1.7cm]{4x23}
            
            \Vertex[a=235, d=1.7cm]{4x31}
            \Vertex[a=210, d=2.3cm]{4x32}
            \Vertex[a=185, d=1.7cm]{4x33}
            
            \foreach \x in {1,2,3}
                \foreach \y in {1,2,3}
                    \Edge(4y)(4x\x\y);
            
            \Edge(4x11)(4x12)
            \Edge(4x11)(4x13)
            \Edge(4x12)(4x13)

            \Edge(4x21)(4x22)
            \Edge(4x21)(4x23)
            \Edge(4x22)(4x23)

            \Edge(4x31)(4x32)
            \Edge(4x31)(4x33)
            \Edge(4x32)(4x33)
        \end{scope}
        \begin{scope}[rotate=120,shift={(0cm, 0cm)}]
            \GraphInit[unit=3,vstyle=Normal]
            \SetVertexNormal[Shape=circle, FillColor=black, MinSize=2pt]
            \tikzset{VertexStyle/.append style = {inner sep = \inners, outer sep = \outers}}
            \SetVertexNoLabel
            \Vertex[Math, x=0,y=0,L={y_1}, Lpos=290, LabelOut, Ldist=3pt]{0y}
            \Vertex[a=-15, d=1.7cm]{0x11}
            \Vertex[a=-30, d=2.3cm]{0x12}
            \Vertex[a=-45, d=1.7cm]{0x13}
            
            \Vertex[a=-87, d=1.7cm]{0x21}
            \Vertex[a=-102, d=2.3cm]{0x22}
            \Vertex[a=-117, d=1.7cm]{0x23}
            
            \Vertex[a=-159, d=1.7cm]{0x31}
            \Vertex[a=-174, d=2.3cm]{0x32}
            \Vertex[a=-189, d=1.7cm]{0x33}
            
            \Vertex[a=-231, d=1.7cm]{0x41}
            \Vertex[a=-246, d=2.3cm]{0x42}
            \Vertex[a=-261, d=1.7cm]{0x43}
            
            \Vertex[a=-303, d=1.7cm]{0x51}
            \Vertex[a=-318, d=2.3cm]{0x52}
            \Vertex[a=-333, d=1.7cm]{0x53}

            \foreach \x in {1,2,3,4,5}
                \foreach \y in {1,2,3}
                    \Edge(0y)(0x\x\y);
                    
                    \Edge(0x11)(0x11)
            \Edge(0x11)(0x12)
            \Edge(0x11)(0x13)
            \Edge(0x12)(0x12)
            \Edge(0x12)(0x13)
            \Edge(0x13)(0x13)
            \Edge(0x21)(0x21)
            \Edge(0x21)(0x22)
            \Edge(0x21)(0x23)
            \Edge(0x22)(0x22)
            \Edge(0x22)(0x23)
            \Edge(0x23)(0x23)
            \Edge(0x31)(0x31)
            \Edge(0x31)(0x32)
            \Edge(0x31)(0x33)
            \Edge(0x32)(0x32)
            \Edge(0x32)(0x33)
            \Edge(0x33)(0x33)
            \Edge(0x41)(0x41)
            \Edge(0x41)(0x42)
            \Edge(0x41)(0x43)
            \Edge(0x42)(0x42)
            \Edge(0x42)(0x43)
            \Edge(0x43)(0x43)
            \Edge(0x51)(0x51)
            \Edge(0x51)(0x52)
            \Edge(0x51)(0x53)
            \Edge(0x52)(0x52)
            \Edge(0x52)(0x53)
            \Edge(0x53)(0x53)
            
        \end{scope}
        \Edge(0y)(1y)
        \Edge(0y)(2y)
        \Edge(0y)(3y)
        \Edge(0y)(4y)

    \end{tikzpicture}
    
    \caption{\textsc{equitable coloring} instance built on Theorem~\ref{thm:blocks} corresponding to the \textsc{bin-packing} instance $A = \{2,2,2,2\}$, $k=3$ and $B = 4$.}
    \label{fig:super_flower}
\end{figure}
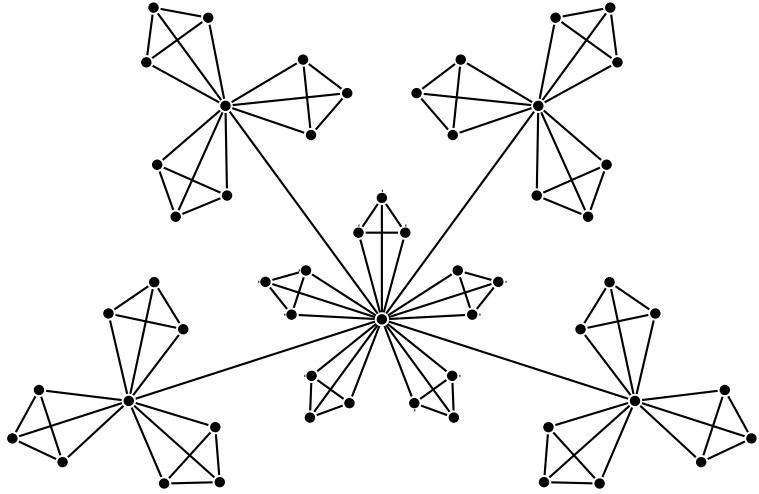

All of our reductions involve the \textsc{bin-packing} problem, which is $\NPH$ in the strong sense~\citep{garey_johnson} and $\W[1]$-$\Hard$ when parameterized by the number of bins~\citep{bin_packing_w1}.
In the general case, the problem is defined as: given a set of positive integers $A = \{a_1, \dots, a_n\}$, called \textit{items}, and two integers $k$ and $B$, can we partition $A$ into $k$ \emph{bins} such that the sum of the elements of each bin is at most $B$?
We shall use a version of \textsc{bin-packing} where each bin sums \emph{exactly} to~$B$.
This second version is equivalent to the first, even from the parameterized point of view; it suffices to add~$kB - \sum_{j \in [n]} a_j$ unitary items to~$A$.
For simplicity, by \textsc{bin-packing} we shall refer to the second version, which we formalize as follows.

\pproblem{bin-packing}{A set of $n$ items $A$ and a bin capacity $B$.}{The number of bins $k$.}{Is there a $k$-partition $\varphi$ of $A$ such that, $\forall i \in [k]$, $\sum_{a_j \in \varphi_i} a_j = B$?}

The idea for the following reductions is to build one gadget for each item~$a_j$ of the given \textsc{bin-packing} instance, perform their disjoint union, and equitably $k$-color the resulting graph.
The color given to the circled vertices in Figure~\ref{fig:flower} control the bin to which the corresponding item belongs to.
Each reduction uses only one of the three gadget types.
Since every gadget is a chordal graph, their treewidth is precisely the size of the largest clique minus one, that is, $k$, which is also the number of desired colors for the built instance of \textsc{equitable coloring}.

\subsection{Disjoint union of Split Graphs}

\begin{definition}
    An $(a,k)$-antiflower is the graph $F_-(a,k) = K_{k-1} \oplus \left(\bigcup_{i \in [a+1]} K_1\right)$, that is, it is the graph obtained after performing the disjoint union of $a+1$ $K_1$'s followed by the join with $K_{k-1}$.
\end{definition}

\begin{theorem}
    \label{thm:dis_split}
    \textsc{equitable coloring} of the disjoint union of split graphs parameterized by the number of colors is $\W[1]$-$\Hard$.
\end{theorem}
\begin{proof}
    Let $\langle A,k,B\rangle$ be an instance of \textsc{bin-packing} and $G$ a graph such that~$G = \bigcup_{j \in [n]} F_-(a_j,k)$.
    Note that $|V(G)| = \sum_{j \in [n]} |F_-(a_j, k)| = \sum_{j \in [n]} k + a_j = nk + kB$.
    Therefore, in any equitable $k$-coloring of~$G$, each color class has~$n + B$ vertices.
    Define~$F_j = F_-(a_j,k)$ and let~$C_j$ be the corresponding~$K_{k-1}$.
    We show that there is an equitable~$k$-coloring~$\psi$ of~$G$ if and only if~$\varphi = \langle A,k,B\rangle$ is a~$\YES$ instance of \textsc{bin-packing}.
    
    Let~$\varphi$ be a solution to \textsc{bin-packing}.
    For each~$a_j \in A$, we do $\psi(C_j) = [k] \setminus \{i\}$ if $a_j \in \varphi_i$.
    We color each vertex of the independent set of~$F_{j}$ with~$i$ and note that all remaining possible proper colorings of the gadget use each color the same number of times.
    Thus, $|\psi_i| = \sum_{j \mid a_j \in \varphi_i} (a_j + 1) + \sum_{j \mid a_j \notin \varphi_i} 1 = \sum_{j \mid a_j \in \varphi_i} (a_j + 1) + \sum_{j \in [n]} 1 - \sum_{j \mid a_j \in \varphi_i} 1 = n + B$.
    
    Now, let~$\psi$ be an equitable $k$-coloring of~$G$.
    Note that~$|\psi_i| = n+B$ and that the independent set of an antiflower is monochromatic.
    For each $j \in [n]$, $a_j \in \varphi_i$ if $i \notin \psi(C_j)$.
    That is, $n + B = |\psi_i| = \sum_{j \mid i \notin C_j} (a_j + 1) + \sum_{j \mid i \in C_j} 1 = \sum_{j \mid i \notin C_j} (a_j + 1) + \sum_{j \in [n]} 1 - \sum_{j \mid i \notin C_j} 1 = \sum_{j \mid i \notin C_j} a_j + n$, from which we conclude that $\sum_{j \mid i \notin C_j} a_j = B$.
\end{proof}

\subsection{Block Graphs}

We now proceed to the parameterized complexity of block graphs.
Conceptually, the proof follows a similar argumentation as the one developed in Theorem~\ref{thm:dis_split}; in fact, we are able to show that even restricting the problem to graphs of diameter at least four is not enough to develop an $\FPT$ algorithm, unless $\FPT =~\W[1]$.
 
\begin{definition}
    An $(a,k)$-flower is the graph $F(a,k) = K_1 \oplus \left(\bigcup_{i \in [a+1]} K_{k-1}\right)$, that is, it is obtained from the union of $a+1$ cliques of size $k-1$ followed by a join with $K_1$.
\end{definition}

\begin{theorem}
    \label{thm:blocks}
    \textsc{equitable coloring} of block graphs of diameter at least four parameterized by the number of colors and treewidth is $\W[1]$-$\Hard$. 
\end{theorem}

\begin{tproof}
    Let $\langle A,k,B\rangle$ be an instance of \textsc{bin-packing}, $\forall k \in [n]$, $F_j = F(a_j, k+1)$, $F_0 = F(B, k+1)$ and, for $j \in \{0\} \cup [n]$, let $y_j$ be the universal vertex of $F_j$.
    Define a graph $G$ such that $V(G) = V\left(\bigcup_{j \in \{0\} \cup [n]} V(F_j)\right)$ and $E(G) = \{y_0y_j \mid j \in [n]\} \cup E\left(\bigcup_{j \in \{0\} \cup [n]} E(F_j)\right)$.
    Looking at Figure~\ref{fig:super_flower}, it is easy to see that any minimum path between a non-universal vertex of $F_a$ and a non-universal vertex of~$F_b$, $a \neq b \neq 0$ has length four.
    We show that $\langle A,k,B\rangle$ is an $\YES$ instance if and only if $G$ is equitably $(k+1)$-colorable.
    \begin{align*}
        |V(G)| &= |V(F_0)| + \sum_{j \in [n]} |V(F_j)| &= &k(B + 1) + 1 + \sum_{j \in [n]} \left(1 + k(a_j + 1)\right)\\
               &= kB + k + n + k^2B + kn + 1           &= &(k+1)(kB + n + 1)\\
    \end{align*}
    
    Given a $k$-partition $\varphi$ of $A$ that solves our instance of \textsc{bin-packing}, we construct a coloring $\psi$ of~$G$ such that $\psi(y_j) = i$ if $a_j \in \varphi_i$ and $\psi(y_0) = k+1$.
    Using a similar argument to the previous theorem, after coloring each $y_j$, the remaining vertices of $G$ are automatically colored.
    For $\psi_{k+1}$, note that $|\psi_{k+1}| = 1 + \sum_{j \in [n]} (a_j + 1) = kB + n + 1 = \frac{|V(G)|}{k+1}$.
    It remains to prove that every other color class $\psi_i$ also has~$\frac{|V(G)|}{k+1}$ vertices.
    
    \begin{align*}
        |\psi_i| &= B + 1 + \sum_{j \mid y_j \notin \psi_i} (a_j + 1) + \sum_{j \mid y_j \in \psi_i} 1 &= &B + 1 + \sum_{j \in [n]} (a_j + 1) - \sum_{j \mid y_j \in \psi_i} a_j\\
                 &= B + 1 + kB + n - B &= &kB + n + 1\\
    \end{align*}
    
    For the converse we take an equitable $(k+1)$-coloring of $G$ and suppose, without loss of generality, that $\psi(y_0) = k+1$ and, consequently, for every other $y_i$, $\psi(y_i) \neq k+1$.
    To build our $k$-partition $\varphi$ of~$A$, we say that $a_j \in \varphi_i$ if $\psi(y_j) = i$.
    The following equalities show that $\sum_{a_j \in \varphi_i} a_j = B$ for every $i$, completing the proof.
    
    \begin{align*}
        |\psi_i|  &= B + 1 + \sum_{j \mid y_j \in \psi_i} 1 + \sum_{j \mid y_j \notin \psi_i} (a_j + 1) &= &B + 1 + \sum_{j \in [n]} (a_j + 1) - \sum_{j \mid y_j \in \psi_i} a_j \\
        kB + n + 1 &= B + 1 + kB + n - \sum_{j \mid y_j \in \psi_i} a_j &\Rightarrow &B =  \sum_{j \mid y_j \in \psi_i} a_j
    \end{align*}
\end{tproof}

\subsection{Interval Graphs without some induced stars}

Before proceeding to our last reduction, we present a polynomial time algorithm to equitably $k$-color a claw-free chordal graph $G$.
To do this, given a partial $k$-coloring $\varphi$ of $G$, denote by $G[\varphi]$ the subgraph of $G$ induced by the vertices colored with $\varphi$, define $\varphi_-$ as the set of colors used $\floor{|V(G[\varphi])|/k}$ times in $\varphi$ and $\varphi_+$ the remaining colors.
If $k$ divides $|V(G[\varphi])|$, we say that $\varphi_+ = \emptyset$.
Our goal is to color~$G$ one maximal clique (say $Q$) at a time and keep the invariant that, the new vertices introduced by $Q$ can be colored a subset of the elements of $L_-$.
To do so, we rely on the fact that, for claw-free graphs, the maximal connected components of the subgraph induced by any two colors form either cycles, which cannot happen since $G$ is chordal, or paths.
By carefully choosing which colors to look at, we find odd length paths that can be greedily recolored to restore our invariant.

\begin{lemma}\label{lem:1}
    There is an $\bigO{n^2}$-time algorithm to equitably $k$-color a claw-free chordal graph or determine that no such coloring exists.
\end{lemma}

\begin{proof}
    We proceed by induction on the number $n$ of vertices of $G$, and show that $G$ is equitably $k$-colorable if and only if its maximum clique has size at most $k$.
    The case $n = 1$ is trivial.
    For general $n$, take one of the leaves of the clique tree of $G$, say $Q$, a simplicial vertex $v \in Q$ and define $G' = G \setminus \{v\}$.
    By the inductive hypothesis, there is an equitable $k$-coloring of $G'$ if and only if $k \geq \omega(G')$.
    If $k < \omega(G')$ or $k < |Q|$, $G$ can't be properly colored.
    
    Now, since $k \geq \omega(G) \geq |Q|$, take an equitable $k$-coloring $\varphi'$ of $G'$ and define $Q' = Q \setminus \{v\}$.
    If $|\varphi'_- \setminus \varphi'(Q')| \geq 1$, we can extend $\varphi'$ to $\varphi$ using one of the colors of $\varphi'_- \setminus \varphi'(Q')$ to greedily color~$v$.
    Otherwise, note that $\varphi'_+ \setminus \varphi'(Q') \neq \emptyset$ because $k \geq \omega(G')$.
    Now, take some color $c \in \varphi'_- \cap \varphi'(Q')$, $d \in \varphi'_+ \setminus \varphi'(Q')$; by our previous observation, we know that $G'[\varphi_c \cup \varphi_d]$ has $C = \{C_1, \dots, C_l\}$ connected components, which in turn are paths.
    Now, take $C_i \in C$ such that $C_i$ has odd length and both endvertices are colored with $d$; said component must exist since $d \in \varphi'_+$ and $c \in \varphi'_-$.
    Moreover, $C_i \cap Q' = \emptyset$, we can swap the colors of each vertex of $C_i$ and then color $v$ with $d$; neither operation makes an edge monochromatic.
    
    As to the complexity of the algorithm, at each step we may need to select $c$ and $d$ -- which takes $\bigO{k}$ time -- construct $C$, find $C_i$ and perform its color swap, all of which take $\bigO{n}$ time.
    Since we need to color $n$ vertices and $k \leq n$, our total complexity is $\bigO{n^2}$.
\end{proof}

The above algorithm was not the first to solve \pname{equitable coloring} for claw-free graphs; this was accomplished by \cite{claw_free_de_werra} which implies that, for any claw-free graph $G$, $\chi_=(G) = \chi_=^*(G) = \chi(G)$.

\begin{theorem}[\cite{claw_free_de_werra}]
    If $G$ is claw-free and $k$-colorable, then $G$ is equitably $k$-colorable.
\end{theorem}

However, \citep{claw_free_de_werra} is not easily accessible, as it is not in any online repository.
Moreover, the given algorithm has no clear time complexity and, as far as we were able to understand the proof, its running time would be $\bigO{k^2n}$, which, for $k = f(n)$, is worse than the algorithm we present in Lemma~\ref{lem:1}.

\begin{definition}
    Let $\mathcal{Q} = \{Q_1, Q_1', \dots, Q_a, Q_a'\}$ be a family of cliques such that $Q_i \simeq Q_i' \simeq K_{k-1}$ and $Y = \{y_1, \dots, y_a\}$ be a set of vertices.
    An $(a,k)$-trem is the graph $H(a,k)$ where $V\left(H(a,k)\right) = \mathcal{Q} \cup Y$ and $E\left(H(a,k)\right) = E\left(\bigcup_{i \in [a]} (Q_i \cup Q_i') \oplus y_i\right) \cup E\left(\bigcup_{i \in [a-1]} y_i \oplus Q_{i+1}\right)$.
\end{definition}

\begin{theorem}
    \label{thm:chordal}
    Let $G$ be a $K_{1,r}$-free interval graph. If $r \geq 4$, \textsc{equitable coloring} of $G$ parameterized by treewidth, number of colors and maximum degree is $\W[1]$-$\Hard$. Otherwise, the problem is solvable in polynomial time.
\end{theorem}

\begin{tproof}
    Once again, let $\langle A,k,B\rangle$ be an instance of \textsc{bin-packing}, define $\forall j \in [n]$, $H_j = H(a_j, k)$ and let $Y_j$ be the set of cut-vertices of $H_j$.
    The graph $G$ is defined as $G = \bigcup_{j \in [n]} V(H_j)$.
    By the definition of an $(a,k)$-trem, we note that the vertices with largest degree are the ones contained in $Y_j \setminus \{y_a\}$, which have degree equal to $3(k-1)$.
    We show that $\langle A,k,B\rangle$ is an $\YES$ instance if and only if $G$ is equitably $k$-colorable, but first note that $|V(G)| = \sum_{j \in [n]} |V(H_j)| = \sum_{j \in [n]} a_j + 2a_j(k-1) = kB + 2(k-1)kB = k(2kB - B)$.
    
    Given a $k$-partition $\varphi$ of $A$ that solves our instance of \textsc{bin-packing}, we construct a coloring $\psi$ of $G$ such that, for each $y \in Y_j$, $\psi(y) = i$ if and only if $a_j \in \varphi_i$.
    Using a similar argument to the other theorems, after coloring each $Y_j$, the remaining vertices of $G$ are automatically colored, and we have  $|\psi_i| = \sum_{j \in \varphi_i} a_j + \sum_{j \notin \varphi_i} 2a_j = B + 2(k-1)B = 2kB - B$.
    
    For the converse we take an equitable $k$-coloring of $G$ and observe that, for every $j \in [n]$, $|\psi(Y_j)| = 1$.
    As such, to build our $k$-partition $\varphi$ of $A$, we say that $a_j \in \varphi_i$ if and only if $\psi(Y_j) = \{i\}$.
    Thus, since $|\psi| = 2kB - B$, we have that $2kB - B = \sum_{j \in \varphi_i} a_j + \sum_{j \notin \varphi_i} 2a_j = \sum_{j \in [n]} a_j - \sum_{j \in \varphi_i} a_j = 2kB - \sum_{j \in \varphi_i} a_j$, from which we conclude that $B = \sum_{j \in \varphi_i} a_j$.
\end{tproof}
\section{Clique Partitioning}

Since \textsc{equitable coloring} is $\W[1]$-$\Hard$ when simultaneously parameterized by many parameters, we are led to investigate a related problem.
Much like \textsc{equitable coloring} is the problem of partitioning~$G$ in $k'$ independent sets of size $\ceil{n/k}$ and $k - k'$ independent sets of size $\floor{n/k}$, one can also attempt to partition $\overline{G}$ in cliques of size $\ceil{n/k}$ or $\floor{n/k}$.
A more general version of this problem is formalized as follows:

\problem{clique partitioning}{A graph $G$ and two positive integers $k$ and $r$.}{Can $G$ be partitioned in $k$ cliques of size $r$ and $\frac{n-rk}{r-1}$ cliques of size $ r - 1$?}

We note that both \textsc{maximum matching} (when $k \geq n/2$) and \textsc{triangle packing} (when $k < n/2$) are particular instances of \textsc{clique partitioning}, the latter being $\FPT$ when parameterized by~$k$~\citep{triangle_packing}.
As such, we will only be concerned when $r \geq 3$.
To the best of our efforts, we were unable to provide an $\FPT$ algorithm for \textsc{clique partitioning} when parameterized by $k$ and $r$, even if we fix $r = 3$.
However, the situation is different when parameterized by the treewidth of $G$, and we obtain an algorithm running in $2^{\tw(G)}n^{\bigO{1}}$ time for the corresponding counting problem, \textsc{\#clique partitioning}.

The key ideas for our bottom-up dynamic programming algorithm are quite straightforward. First, cliques are formed only when building the tables for forget nodes. Second, for join nodes, we can safely consider only the combination of two partial solutions that have empty intersection on the covered vertices (i.e. that have already been assigned to some clique). Finally, both join and forget nodes can be computed using fast subset convolution~\citep{fourier_mobius}.
For each node $x \in \td{T}$, our algorithm builds the table $f_x(S, k')$, where each entry is indexed by a subset $S \subseteq B_x$ that indicates which vertices of $B_x$ have already been covered, an integer $k'$ recording how many cliques of size $r$ have been used, and stores how many partitions exist in $G_x$ such that only $B_x \setminus S$ is yet uncovered.
If an entry is inconsistent (e.g. $S \nsubseteq B_x$), we say that $f(S, K') = 0$.

\begin{theorem}
    \label{thm:clique_part}
    There is an algorithm that, given a nice tree decomposition of an $n$-vertex graph $G$ of width $\tw$, computes the number of partitions of $G$ in $k$ cliques of size $r$ and $\frac{n - rk}{r-1}$ cliques of size $r-1$ in time $\bigOs{2^{\tw}}$ time.
\end{theorem}

\begin{proof}
    \textit{Leaf node:} Take a leaf node $x \in \td{T}$ with $B_x = \emptyset$. 
    Since the only one way of covering an empty graph is with zero cliques, we compute $f_x$ with:
    \begin{equation*}
        \centering
        \hfill f_x(S, k') =
        \begin{cases}
            1, \text{ if } k' = 0 \text{ and } S = \emptyset \text{;}\\
            0, \text{ otherwise.}\\
        \end{cases}
    \end{equation*}
    
    \emph{Introduce node:} Let $x$ be a an introduce node, $y$ its child and $v \in B_x \setminus B_y$.
    Due to our strategy, introduce nodes are trivial to solve; it suffices to define $f_x(S, k') = f_y(S, k')$. If $v \in S$, we simply define~$f_x(S, k') = 0$.
    
    \emph{Forget node:} For a forget node $x$ with child $y$ and forgotten vertex $v$, we formulate the computation of~$f_x(S, k')$ as the subset convolution of two functions as follows:
    
    \begin{align*}
        \centering
        f_x(S, k') &= f_y(S \cup \{v\}, k') + \sum_{A \subseteq S} f_y(S \setminus A, k' - 1)g_r(A, v) + \sum_{A \subseteq S} f_y(S \setminus A, k')g_{r-1}(A, v)\\
        g_l(A, v) &=
            \begin{cases}
                1, \text{ if $A$ is a clique of size $l$ contained in $N[v]$ and $v \in A$;}\\
                0, \text{otherwise.}
            \end{cases}
    \end{align*}
    
    The above computes, for every $S \subseteq B_x$ and every clique $A$ (that contains $v$) of size $r$ or $r - 1$ contained in $N[v] \cap B_y \cap S$, if $S \setminus A$ and some $k''$ is a valid entry of $f_y$, or if $v$ had been previously covered by another clique (first term of the sum).
    Directly computing the last two terms of the equation, for each pair~$(S, k')$, yields a total running time of the order of $ \sum_{|S| = 0}^\tw \binom{\tw}{|S|}2^{|S|} = (1+2)^\tw = 3^\tw$ for each forget node.
    However, using the fast subset convolution technique described by~\cite{fourier_mobius}, we can compute the above equation in time $\bigOs{2^{|B_x|}} = \bigOs{2^{\tw}}$.
    
    Correctness follows directly from the hypothesis that $f_y$ is correctly computed and that, for every $A \subseteq B_x$, $g_r(A, v)g_{r-1}(A, v) = 0$.
    For the running time, we can pre-compute both $g_r$ and $g_{r-1}$ in~$\bigO{2^\tw r^2}$, so their values can be queried in $\bigO{1}$ time.
    As such, each forget node takes $\bigO{2^\tw\tw^3k}$ time, since we can compute the subset convolutions of $f_y * g_r$ and $f_y * g_{r-1}$ in $\bigO{2^\tw\tw^3}$ time each.
    The additional factor of $k$ comes from the second coordinate of the table index.
    
    \emph{Join node:} Take a join node $x$ with children $y$ and $z$.
    Since we want to partition our vertices, the cliques we use in $G_y$ and $G_z$ must be completely disjoint and, consequently, the vertices of $B_x$ covered in $B_y$ and~$B_z$ must also be disjoint.
    As such, we can compute $f_x$ through the equation:
    
    \begin{equation*}
        f_x(S, k') = \sum_{k_y + k_z = k'} \sum_{A \subseteq S} f_y(A, k_y)f_z(S \setminus A, k_z)
    \end{equation*}
    
    Note that we must sum over the integer solutions of the equation $k_y + k_z = k'$ since we do not know how the cliques of size $r$ are distributed in $G_x$.
    To do that, we compute the subset convolution $f_y(\cdot, k_y) * f_z(\cdot, k_z)$.
    The time complexity of $\bigO{2^\tw\tw^3k^2}$ follows directly from the complexity of the fast subset convolution algorithm, the range of the outermost sum and the range of the second parameter of the table index.
    
    For the root $x$, we have $f_x(\emptyset, k) \neq 0$ if and only if $G_x = G$ can be partitioned in $k$ cliques of size $r$ and the remaining vertices in cliques of size $r-1$.
    Since our tree decomposition has $\bigO{n\tw}$ nodes, our algorithm runs in time $\bigO{2^\tw\tw^4k^2n}$.
    
    To recover a solution given the tables $f_x$, start at the root node with $S = \emptyset$, $k' = k$ and let $\mathcal{Q} = \emptyset$ be the cliques in the solution.
    We shall recursively extend $\mathcal{Q}$ in a top-down manner, keeping track of the current node $x$, the set of vertices $S$ and the number $k'$ of $K_r$'s used to cover $G_x$.
    Our goal is to keep the invariant that $f_x(S, k') \neq 0$.
    
    \emph{Introduce node:} Due to the hypothesis that $f_x(S, k') \neq 0$ and the way that $f_x$ is computed, it follows that $f_y(S, k') \neq 0$.
        
    \emph{Forget node:} Since the current entry is non-zero, there must be some $A \subseteq S$ such that exactly one of the products $f_y(S \setminus A, k' - 1)g_r(A, v)$, $f_y(S \setminus A, k')g_{r-1}(A, v)$ is non-zero and, in fact, any such $A$ suffices.
    To find this subset, we can iterate through $2^S$ in $\bigO{2^\tw}$ time and test both products to see if any of them is non-zero.
    Note that the chosen $A \cup \{v\}$ will be a clique of size either $r$ or $r-1$, and thus, we can set $\mathcal{Q} \gets \mathcal{Q} \cup \{A \cup \{v\}\}$.
        
    \emph{Join node:} The reasoning for join nodes is similar to forget nodes, however, we only need to determine which states to look at in the child nodes.
    That is, for each integer solution to $k_y + k_z = k'$ and for each $A \subseteq S$, we check if both $f_y(A, k_y)f_z(S \setminus A, k_z)$ is non-zero; in the affirmative, we compute the solution for both children with the respective entries.
    Any such triple $(A, k_y, k_z)$ that satisfies the condition suffices.
    
    Clearly, retrieving the solution takes $\bigO{2^\tw k}$ time per node, yielding a running time of $\bigOs{2^\tw}$.
\end{proof}

\begin{corollary}
    Equitable coloring is $\FPT$ when parameterized by the treewidth of the complement graph.
\end{corollary}

\section{Conclusions}

In this work, we investigated the \pname{equitable coloring} problem.
We developed novel parameterized reductions from \pname{bin-packing}, which is $\W[1]$-$\Hard$ when parameterized by number of bins.
These reductions showed that \pname{equitable coloring} is $\W[1]-\Hard$ in three more cases: (i) if we restrict the problem to block graphs and parameterize by the number of colors, treewidth and diameter; (ii) on the disjoint union of split graphs, a case where the connected case is polynomial; (iii) \pname{equitable coloring} of $K_{1,r}$ interval graphs, for any $r \geq 4$, remains hard even if we parameterize by the number of colors, treewidth and maximum degree.
This, along with a previous result by~\cite{claw_free_de_werra}, establishes a dichotomy based on the size of the largest induced star: for $K_{1,r}$-free graphs, the problem is solvable in polynomial time if $r \leq 2$, otherwise it is $\W[1]-\Hard$.
These results significantly improve the ones by~\cite{colorful_treewidth} through much simpler proofs and in very restricted graph classes.

Since the problem remains hard even for many natural parameterizations, we resorted to a more exotic one -- the treewidth of the complement graph.
By applying standard dynamic programming techniques on tree decompositions and the fast subset convolution machinery of~\cite{fourier_mobius}, we obtain an $\FPT$ algorithm when parameterized by the treewidth of the complement graph.

Natural future research directions include the identification and study of other uncommon parameters that may aid in the design of other $\FPT$ algorithms.
Revisiting \pname{clique partitioning} when parameterized by $k$ and $r$ is also of interest, since its a related problem to \pname{equitable coloring} and the complexity of its natural parameterization is yet unknown.

\acknowledgements
We are grateful to the two anonymous reviewers whose comments/suggestions helped us to improve and clarify the paper.

\bibliographystyle{abbrvnat}

\begin{thebibliography}{0}
\providecommand{\natexlab}[1]{#1}
\providecommand{\url}[1]{\texttt{#1}}
\expandafter\ifx\csname urlstyle\endcsname\relax
  \providecommand{\doi}[1]{doi: #1}\else
  \providecommand{\doi}{doi: \begingroup \urlstyle{rm}\Url}\fi

\end{thebibliography}


\begin{thebibliography}{14}
	\providecommand{\natexlab}[1]{#1}
	\providecommand{\url}[1]{\texttt{#1}}
	\expandafter\ifx\csname urlstyle\endcsname\relax
	\providecommand{\doi}[1]{doi: #1}\else
	\providecommand{\doi}{doi: \begingroup \urlstyle{rm}\Url}\fi
	
	\bibitem[Baker and Coffman(1996)]{mutual_exclusion_scheduling}
	B.~S. Baker and E.~G. Coffman.
	\newblock Mutual exclusion scheduling.
	\newblock \emph{Theoretical Computer Science}, 162\penalty0 (2):\penalty0
	225--243, 1996.
	
	\bibitem[Bj{\"o}rklund et~al.(2007)Bj{\"o}rklund, Husfeldt, Kaski, and
	Koivisto]{fourier_mobius}
	A.~Bj{\"o}rklund, T.~Husfeldt, P.~Kaski, and M.~Koivisto.
	\newblock Fourier meets m{\"o}bius: fast subset convolution.
	\newblock In \emph{Proceedings of the thirty-ninth annual ACM symposium on
		Theory of computing}, pages 67--74. ACM, 2007.
	
	\bibitem[Bodlaender and Fomin(2005)]{equitable_treewidth}
	H.~L. Bodlaender and F.~V. Fomin.
	\newblock \emph{Equitable colorings of bounded treewidth graphs}, volume 349,
	pages 22 -- 30.
	\newblock 2005.
	\newblock \doi{https://doi.org/10.1016/j.tcs.2005.09.027}.
	\newblock Graph Colorings.
	
	\bibitem[Bodlaender and Jansen(1995)]{graph_partitioning1}
	H.~L. Bodlaender and K.~Jansen.
	\newblock Restrictions of graph partition problems. part {I}.
	\newblock \emph{Theoretical Computer Science}, 148\penalty0 (1):\penalty0 93 --
	109, 1995.
	\newblock ISSN 0304-3975.
	
	\bibitem[Chen et~al.(1996)Chen, Ko, and Lih]{equitable_split}
	B.-L. Chen, M.-T. Ko, and K.-W. Lih.
	\newblock \emph{Equitable and $m$-bounded coloring of split graphs}, pages
	1--5.
	\newblock Springer Berlin Heidelberg, Berlin, Heidelberg, 1996.
	
	\bibitem[de~Werra(1985)]{claw_free_de_werra}
	D.~de~Werra.
	\newblock Some uses of hypergraphs in timetabling.
	\newblock \emph{Asia-Pacific Journal of Operational Research}, 2\penalty0
	(1):\penalty0 2--12, 1985.
	
	\bibitem[Downey and Fellows(2013)]{downey_fellows}
	R.~G. Downey and M.~R. Fellows.
	\newblock \emph{Fundamentals of parameterized complexity}, volume~4.
	\newblock Springer, 2013.
	
	\bibitem[Fellows et~al.(2005)Fellows, Heggernes, Rosamond, Sloper, and
	Telle]{triangle_packing}
	M.~Fellows, P.~Heggernes, F.~Rosamond, C.~Sloper, and J.~A. Telle.
	\newblock Finding k disjoint triangles in an arbitrary graph.
	\newblock In J.~Hromkovi{\v{c}}, M.~Nagl, and B.~Westfechtel, editors,
	\emph{Graph-Theoretic Concepts in Computer Science}, pages 235--244, Berlin,
	Heidelberg, 2005. Springer Berlin Heidelberg.
	
	\bibitem[Fellows et~al.(2011)Fellows, Fomin, Lokshtanov, Rosamond, Saurabh,
	Szeider, and Thomassen]{colorful_treewidth}
	M.~R. Fellows, F.~V. Fomin, D.~Lokshtanov, F.~Rosamond, S.~Saurabh, S.~Szeider,
	and C.~Thomassen.
	\newblock On the complexity of some colorful problems parameterized by
	treewidth.
	\newblock \emph{Information and Computation}, 209\penalty0 (2):\penalty0 143 --
	153, 2011.
	
	\bibitem[Garey and Johnson(1979)]{garey_johnson}
	M.~R. Garey and D.~S. Johnson.
	\newblock \emph{Computers and Intractability: A Guide to the Theory of
		NP-Completeness}.
	\newblock W. H. Freeman \&amp; Co., New York, NY, USA, 1979.
	
	\bibitem[Jansen et~al.(2013)Jansen, Kratsch, Marx, and
	Schlotter]{bin_packing_w1}
	K.~Jansen, S.~Kratsch, D.~Marx, and I.~Schlotter.
	\newblock Bin packing with fixed number of bins revisited.
	\newblock \emph{Journal of Computer and System Sciences}, 79\penalty0
	(1):\penalty0 39 -- 49, 2013.
	
	\bibitem[Jarvis and Zhou(2001)]{equitable_trees}
	M.~Jarvis and B.~Zhou.
	\newblock Bounded vertex coloring of trees.
	\newblock \emph{Discrete Mathematics}, 232\penalty0 (1-3):\penalty0 145--151,
	2001.
	
	\bibitem[Lih(2013)]{equitable_survey}
	K.-W. Lih.
	\newblock Equitable coloring of graphs.
	\newblock In \emph{Handbook of combinatorial optimization}, pages 1199--1248.
	Springer, 2013.
	
	\bibitem[Lonc(1992)]{chain_antichain}
	Z.~Lonc.
	\newblock \emph{On complexity of some chain and antichain partition problems},
	pages 97--104.
	\newblock Springer Berlin Heidelberg, Berlin, Heidelberg, 1992.
	
\end{thebibliography}

\end{document}